\documentclass[a4paper,USenglish,cleveref,numberwithinsect]{lipics-v2019}

\graphicspath{{./figures/}}
\title{A 4-Approximation of the $\frac{2\pi}{3}$-MST}

\author{Stav Ashur}
{Department of Computer Science, Ben-Gurion University of the Negev, Israel}
{stavshe@post.bgu.ac.il}
{}
{}
\author {Matthew J. Katz}
{Department of Computer Science, Ben-Gurion University of the Negev, Israel}
{matya@cs.bgu.ac.il}
{}
{Partially supported by grant 1884/16 from the Israel Science Foundation.}
\authorrunning{S. Ashur and M. Katz}

\Copyright{Stav Ashur and Matthew Katz}



\EventEditors{}
\EventNoEds{2}
\EventLongTitle{}
\EventShortTitle{}
\EventAcronym{}
\EventYear{2021}
\EventDate{}
\EventLocation{}
\EventLogo{}
\SeriesVolume{}
\ArticleNo{}

\nolinenumbers 

\usepackage{cite} 
\usepackage{amssymb}
\usepackage{graphicx}
\usepackage{algorithm}
\usepackage[noend]{algorithmic}
\usepackage{microtype}
\usepackage{mathtools} 
\usepackage{amsthm}
\usepackage{enumerate}
\usepackage{wrapfig}
\usepackage{multirow}
\usepackage{nicefrac}
\usepackage[colorinlistoftodos,prependcaption,textsize=tiny]{todonotes}

\usepackage{url}   

\definecolor{forestgreen}{rgb}{0.13, 0.55, 0.13}

\newtheorem{observation}[theorem]{Observation}

\def\eps{{\varepsilon}}

\newcommand{\upt}[2]{\mathcal{P}_{\{#1,#2\}}}
\newcommand{\pt}[2]{\mathcal{P}_{#1,#2}}
\newcommand{\spt}[3]{\mathcal{P}_{#1,#2}^{#3}}

\date{}

\newcommand{\old}[1]{{{}}}

\newcommand{\stav}[1]{\todo[linecolor=red,backgroundcolor=blue!10,bordercolor=red]{\textbf{Stav: }#1}}

\old{

\documentclass{article}
\usepackage[utf8]{inputenc}
\usepackage{amsthm}
\usepackage[english]{babel}
\usepackage{graphicx}
\usepackage{wrapfig}
\usepackage{subcaption}
\usepackage{xcolor}
\usepackage{amsmath} 
\usepackage{bbm}
\usepackage{amsfonts}
\usepackage{cite}
\usepackage{csquotes}

\usepackage[pdftex, plainpages = false, pdfpagelabels, 
				bookmarks=false,
				bookmarksopen = true,
				bookmarksnumbered = true,
				breaklinks = true,
				linktocpage,
				pagebackref,
				colorlinks = true,  
				linkcolor = blue,
				urlcolor  = blue,
				citecolor = red,
				anchorcolor = green,
				hyperindex = true,
				hyperfigures
				]{hyperref}

\usepackage[nameinlink]{cleveref}

\graphicspath{{./figures/}}

\newcommand{\upt}[2]{\mathcal{P}_{\{#1,#2\}}}
\newcommand{\pt}[2]{\mathcal{P}_{#1,#2}}
\newcommand{\spt}[3]{\mathcal{P}_{#1,#2}^{#3}}

\newtheorem{theorem}{Theorem}
\newtheorem{corollary}{Corollary}[theorem]

\newtheorem{claim}[theorem]{Claim}

\newtheorem{definition}[theorem]{definition}

\newcommand{\old}[1]{{{}}}

\newcommand{\stav}[1]{\todo[linecolor=red,backgroundcolor=blue!10,bordercolor=red]{\textbf{Stav: }#1}}

\select@language{english}

\title{alpha Bounded Spanning Trees}
\author{stavshe }
\date{}
}

\begin{document}

\maketitle

\begin{abstract}
	Bounded-angle (minimum) spanning trees were first introduced in the context of wireless networks with directional antennas. They are reminiscent of bounded-degree spanning trees, which have received significant attention. 
	Let $P = \{p_1,\ldots,p_n\}$ be a set of $n$ points in the plane, let $\Pi$ be the polygonal path $(p_1,\ldots,p_n)$, and let $0 < \alpha < 2\pi$ be an angle. An $\alpha$-spanning tree ($\alpha$-ST) of $P$ is a spanning tree of the complete Euclidean graph over $P$, with the following property: For each vertex $p_i \in P$, the (smallest) angle that is spanned by all the edges incident to $p_i$ is at most $\alpha$. An $\alpha$-minimum spanning tree ($\alpha$-MST) is an $\alpha$-ST of $P$ of minimum weight, where the weight of an $\alpha$-ST is the sum of the lengths of its edges. In this paper, we consider the problem of computing an $\alpha$-MST, for the important case where $\alpha = \frac{2\pi}{3}$. We present a simple 4-approximation algorithm, thus improving upon the previous results of Aschner and Katz and Biniaz et al., who presented algorithms with approximation ratios 6 and $\frac{16}{3}$, respectively.
	
	In order to obtain this result, we devise a simple $O(n)$-time algorithm for constructing a $\frac{2\pi}{3}$-ST\, ${\cal T}$ of $P$, such that ${\cal T}$'s weight is at most twice that of $\Pi$ and, moreover, ${\cal T}$ is a 3-hop spanner of $\Pi$. This latter result is optimal in the sense that for any $\eps > 0$ there exists a polygonal path for which every $\frac{2\pi}{3}$-ST has weight greater than $2-\eps$ times the weight of the path. 
\end{abstract}

\section{Introduction}

Let $P = \{p_1,\ldots,p_n\}$ be a set of $n$ points in the plane. An $\alpha$-spanning tree ($\alpha$-ST) of $P$, for an angle $0 < \alpha < 2\pi$, is a spanning tree of the complete Euclidean graph over $P$, with the following property: For each vertex $p_i \in P$, the (smallest) angle that is spanned by all the edges incident to $p_i$ is at most $\alpha$ (see Figure~\ref{fig:120-ST}). An $\alpha$-minimum spanning tree ($\alpha$-MST) is then an $\alpha$-ST of $P$ of minimum weight, where the weight of an $\alpha$-ST is the sum of the lengths of its edges.

\begin{figure}[ht]
	\centering
	\includegraphics[scale=0.6, page=1]{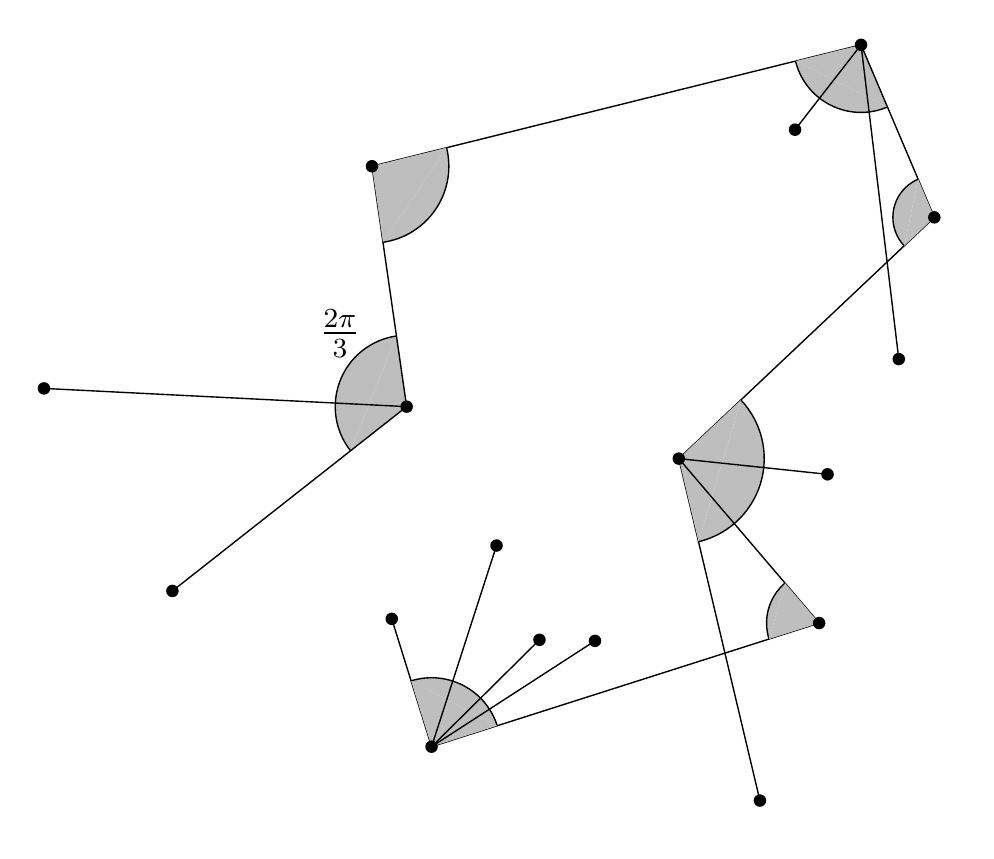}
	\caption{A $\frac{2\pi}{3}$-ST.}
	\label{fig:120-ST}
\end{figure} 

Since there always exists a MST of $P$ in which the degree of each vertex is at most 5, the interesting range for $\alpha$ is $(0,\frac{8\pi}{5})$. The concept of bounded-angle (minimum) spanning tree (i.e., of an $\alpha$-(M)ST) was introduced by Aschner and Katz~\cite{AschnerK17}, who arrived at it through the study of wireless networks with directional antennas. However, it is interesting in its own right. The study of bounded-angle (minimum) spanning trees is also related to the study of bounded-degree (minimum) spanning trees, which received considerable attention (see, e.g.,~\cite{PapadimitriouV84, Chan04,FeketeKKRY97,JothiR09,KhullerRY96}). (A degree-$k$ ST, is a spanning tree in which the degree of each vertex is at most $k$, and a degree-$k$ MST is a degree-$k$ ST of minimum weight.)

It is easy to see that an $\alpha$-ST of $P$, for $\alpha < \frac{\pi}{3}$, does not always exist; think, for example, of the corners of an equilateral triangle. On the other hand, it is known (see~\cite{AckermanGP13,AichholzerHHHPSSV13,CarmiKLR11}) that for any $\alpha \ge \frac{\pi}{3}$, there always exists an $\alpha$-ST of $P$. 

The next natural question is what is the status of the problem of computing an $\alpha$-MST, for a given `typical' angle $\alpha$. Aschner and Katz~\cite{AschnerK17} proved that (at least) for $\alpha = \pi$ and for $\alpha = \frac{2\pi}{3}$ the problem is NP-hard, and, therefore, it calls for efficient approximation algorithms.  

Obviously, the weight of an $\alpha$-MST of $P$, for any angle $\alpha$, is at least the weight of an MST of $P$, so if we develop an algorithm for constructing an $\alpha$-ST, for some angle $\alpha$, and prove that the weight of the trees constructed by the algorithm never exceeds some constant $c$ times the weight of the corresponding MSTs, then we have a $c$-approximation algorithm for computing an $\alpha$-MST. Aschner et al.~\cite{AschnerKM13} showed that this approach is relevant only if $\alpha \ge \frac{\pi}{2}$, since for any $\alpha < \frac{\pi}{2}$, there exists a set of points for which the ratio between the weights of the $\alpha$-MST and the MST is $\Omega(n)$. 

In this paper, we focus on the important case where $\alpha = \frac{2\pi}{3}$. That is, we are interested in an algorithm for computing a `good' approximation of $\frac{2\pi}{3}$-MST, where by good we mean that the weight of the output $\frac{2\pi}{3}$-ST is not much larger than that of an MST (and thus of a $\frac{2\pi}{3}$-MST). Aschner and Katz~\cite{AschnerK17} presented a 6-approximation algorithm for the problem. Subsequently, Biniaz et al.~\cite{BiniazBLM20} described an improved $\frac{16}{3}$-approximation algorithm. In this paper, we manage to reduce the approximation ratio to 4, by taking a completely different approach than the two previous algorithms. 

Most of the paper is devoted to proving Theorem~\ref{thm:path2approx}, which is of independent interest. Our main result, i.e., the 4-approximation algorithm, is obtained as an easy corollary of this theorem. Let $\Pi$ denote the polygonal path $(p_1,...,p_n)$. Then, Theorem~\ref{thm:path2approx} states that one can construct a $\frac{2\pi}{3}$-ST\, ${\cal T}$ of $P$, such that (i) the weight of ${\cal T}$, $\omega({\cal T})$, is at most $2\omega(\Pi)$, and (ii) ${\cal T}$ is a 3-hop spanner of $\Pi$ (i.e., if there is an edge between $p$ and $q$ in $\Pi$, then there is a path consisting of at most 3 edges between $p$ and $q$ in ${\cal T}$). 
Notice that 2 is the best approximation ratio that one can hope for, since Biniaz et al.~\cite{BiniazBLM20} showed that for any $\alpha < \pi$, the weight of an $\alpha$-MST of a set of $n$ points on the line, such that the distance between consecutive points is 1, is at least $2n-3$, whereas the weight of an MST is clearly $n-1$. (This lower bound is also mentioned without a proof in~\cite{AschnerK17}.)

We prove Theorem~\ref{thm:path2approx} by presenting an $O(n)$-time algorithm for constructing ${\cal T}$ and proving its correctness.
The algorithm is very simple and easy to implement, but arriving at it and proving its correctness is far from trivial. One approach for constructing a $\frac{2\pi}{3}$-ST of $P$ is to assign to each vertex of $P$ an orientation, where an orientation of a vertex $p$ is a cone of angle $\frac{2\pi}{3}$ with apex at $p$. The assignment of orientations induces a \emph{transmission} graph $G$ (over $P$), where $\{p_i,p_j\}$ is an edge of $G$ if and only if $p_j$ is in $p_i$'s cone and $p_i$ is in $p_j$'s cone. Now, if $G$ is connected, then by computing a minimum spanning tree of $G$ one obtains a $\frac{2\pi}{3}$-ST of $P$. The challenge is of course to determine the orientations of the vertices, so that $G$ is connected and the weight of a minimum spanning tree of $G$ is bounded by a small constant times $\omega(\Pi)$.  

Next, we describe some of the ideas underlying our algorithm for constructing ${\cal T}$. 
Assume for simplicity that $n$ is even and consider the sequence of edges $X$ obtained from $\Pi$ by removing all the edges at even position (i.e., by removing the edges $\{p_2,p_3\},\{p_4,p_5\},\ldots$). For each edge $e=\{p,q\} \in X$, we consider the partition of the plane into four regions induced by $e$, see Figure~\ref{fig:2_points}. This partition determines for each of $e$'s vertices three `allowable' orientations, see Figure~\ref{fig:basic_positions}. Our algorithm assigns to each vertex of $\Pi$ one of its three allowable orientations, such that the resulting transmission graph $G$ contains the edges in $X$ and at least one edge between any two adjacent edges in $X$. Finally, by keeping only the edges in $X$ and a single edge between any two adjacent edges, we obtain ${\cal T}$. The novelty of the algorithm is in the way it assigns the orientations to the vertices to ensure that the resulting graph satisfies these conditions.

We now discuss the two previous results on computing an approximation of a $\frac{2\pi}{3}$-MST of $P$, and some of the related results.
The first stage in the previous algorithms, as well as in the new one, is to compute a MST of $P$, $\mathrm{MST}(P)$, and from it a spanning path $\Pi$ of $P$ of weight at most $2\omega(\mathrm{MST}(P))$. ($\Pi$ is obtained by listing the vertices of $P$ through an in-order traversal of $\mathrm{MST}(P)$, where a vertex is added to the list when it is visited for the first time.) The algorithm of Aschner and Katz~\cite{AschnerK17} operates on the path $\Pi$. It constructs a $\frac{2\pi}{3}$-ST of $P$ from $\Pi$ of weight at most $3\omega(\Pi)$, and thus of weight at most $6\omega(\mathrm{MST}(P))$. The algorithm of Biniaz et al.~\cite{BiniazBLM20} can operate only on non-crossing paths, so it first transforms $\Pi$ to a non-crossing path $\Pi'$ (by iteratively flipping crossing edges), such that $\omega(\Pi') \le \omega(\Pi)$. Then, it constructs a $\frac{2\pi}{3}$-ST of $P$ from $\Pi'$ of weight at most $\frac{8}{3}\omega(\Pi')$, and thus of weight at most $\frac{16}{3}\omega(\mathrm{MST}(P))$. The new algorithm operates directly on $\Pi$. It constructs a $\frac{2\pi}{3}$-ST of $P$ from $\Pi$ of weight at most $2\omega(\Pi)$, and thus of weight at most $4\omega(\mathrm{MST}(P))$.  

Notice that 4 is the best approximation ratio possible, for any such two-stage algorithm, provided the stages are independent. This is true since (i) Fekete et al.~\cite{FeketeKKRY97} showed that for any $\eps > 0$ there exists a point set for which any spanning path has weight at least $2 - \eps$ times the weight of an MST, and (ii) as mentioned above, for any $\eps > 0$, there exists a point set and a corresponding spanning path for which any $\frac{2\pi}{3}$-ST has weight at least $2 - \eps$ times the weight of the path. 

As for other values of $\alpha$, Aschner and Katz~\cite{AschnerK17} presented a 16-approximation algorithm for computing a $\frac{\pi}{2}$-MST of $P$.
The best known approximations of the degree-$k$ MST, for $k = 2,3,4$, imply a 2-approximation of the $\pi$-MST, a 1.402-approximation of the $\frac{4\pi}{3}$-MST~\cite{Chan04}, and a 1.1381-approximation of the $\frac{3\pi}{2}$-MST~\cite{Chan04,JothiR09}.

\section{Preliminaries}

\begin{figure}[ht]
	\centering
	\includegraphics[scale=0.6, page=1]{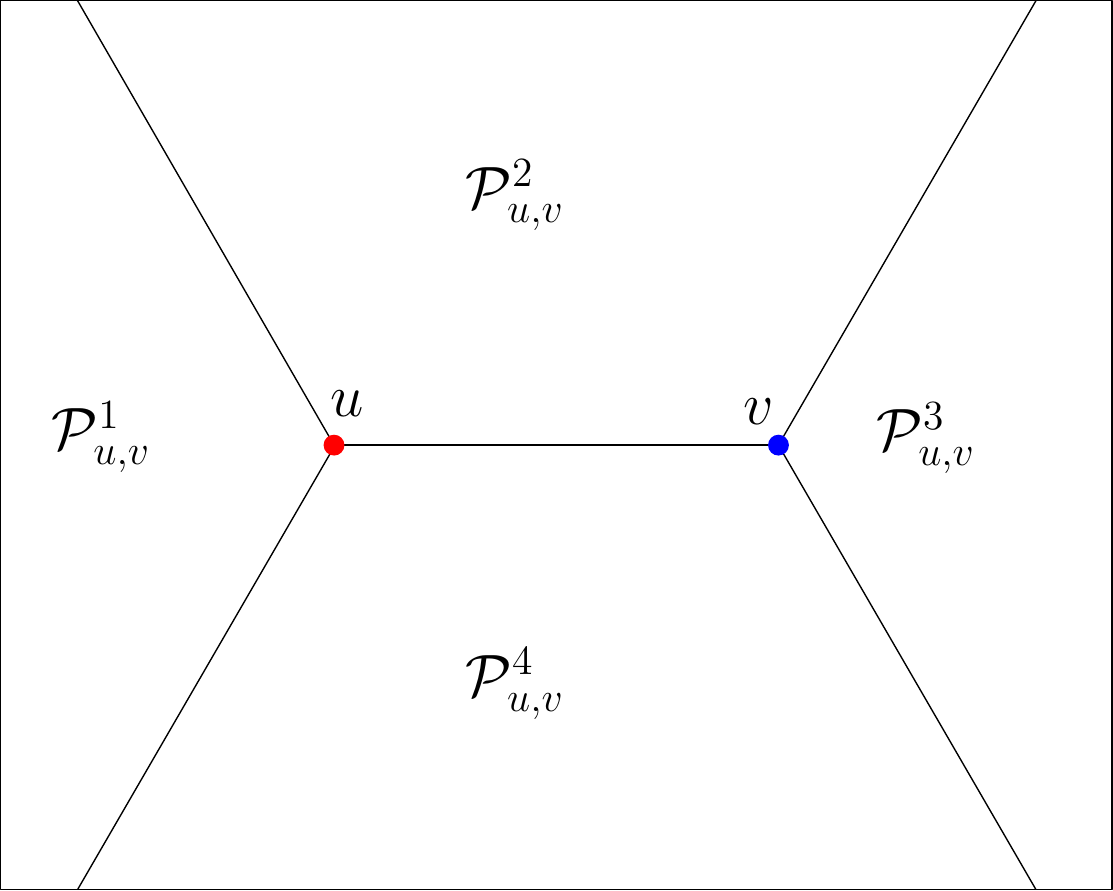}
	\caption{The partition of the plane $\pt{u}{v}$ induced by the ordered pair of points $(u,v)$.}
	\label{fig:2_points}
\end{figure}

\begin{definition}
Any ordered pair $(u,v)$ of points in the plane, induces a partition of the plane into four regions, which we denote by $\pt{u}{v}$; see Figure~\ref{fig:2_points}. We denote the four regions by $\spt{u}{v}{1}$, $\spt{u}{v}{2}$, $\spt{u}{v}{3}$, and $\spt{u}{v}{4}$, as depicted in Figure~\ref{fig:2_points}. Notice that the partitions $\pt{u}{v}$ and $\pt{v}{u}$ are identical, where $\spt{u}{v}{1} = \spt{v}{u}{3}$, $\spt{u}{v}{2} = \spt{v}{u}{4}$, etc. Sometimes, we prefer to consider the points $u$ and $v$ as an unordered pair of points, in which case we denote the partition induced by them as $\upt{u}{v}$. In $\upt{u}{v}$, we distinguish between the two \emph{side} regions, which are $\spt{u}{v}{1}$ and $\spt{u}{v}{3}$ (alternatively, $\spt{v}{u}{3}$ and $\spt{v}{u}{1}$), and the two \emph{center} regions, which are $\spt{u}{v}{2}$ and $\spt{u}{v}{4}$ (alternatively, $\spt{v}{u}{4}$ and $\spt{v}{u}{2}$). 
\end{definition}

\begin{figure}[ht]
	\begin{subfigure}{0.48\textwidth}
		\includegraphics[width=\textwidth, page=2]{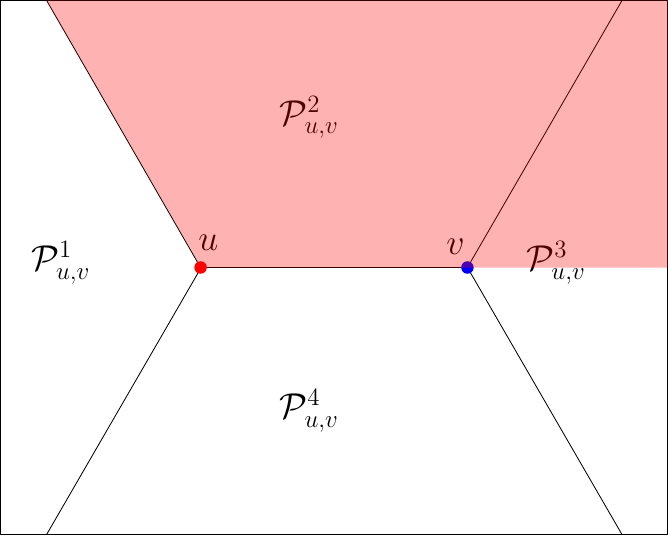}
		\subcaption{The orientation $u_v^u$.}
	\end{subfigure}
	\hfil
	\centering
	\begin{subfigure}{0.48\textwidth}
		\includegraphics[width=\textwidth, page=4]{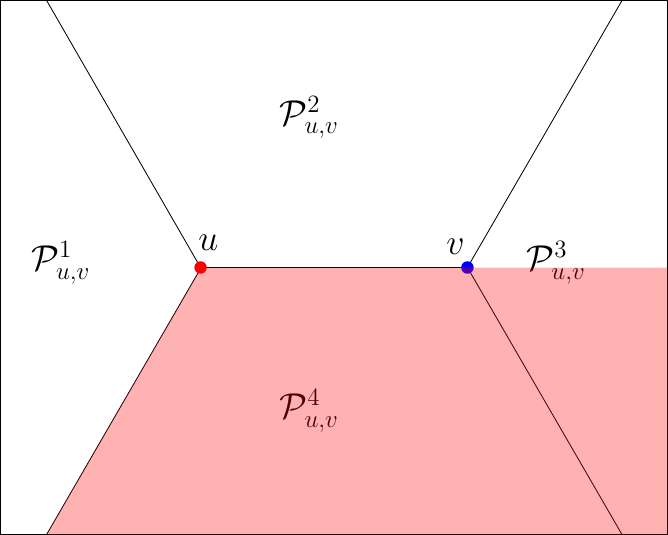}
		\subcaption{The orientation $u_v^d$.}
	\end{subfigure}
	
	\vspace{0.2cm}
	\begin{subfigure}{0.48\textwidth}
		\includegraphics[width=\textwidth, page=3]{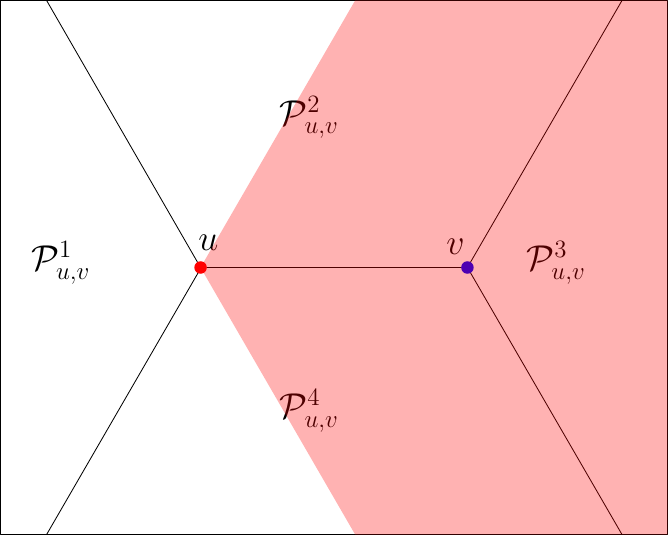}
		\subcaption{The orientation $u_v^c$.}
	\end{subfigure}
		
	\caption{The three basic orientations of $u$ with respect to $v$.}
	\label{fig:basic_positions}
\end{figure}
The \emph{orientation} of a point $u$, is the orientation of a $\frac{2\pi}{3}$-cone with apex at $u$; we refer to this cone as the \emph{transmission} cone of $u$. In the following definition, we define the three \emph{basic} orientations of $u$ with respect to another point $v$, based on $\pt{u}{v}$; see Figure~\ref{fig:basic_positions}.

\begin{definition}For a pair of points $u$ and $v$, the three basic orientations of $u$ with respect to $v$ are:
\begin{description}
	\item{$\mathbf{u_v^u}$:}
		The only orientation of $u$, such that $\spt{u}{v}{2}$ is fully contained in the transmission cone of $u$,
	\item{$\mathbf{u_v^d}$:}
		The only orientation of $u$, such that $\spt{u}{v}{4}$ is fully contained in the transmission cone of $u$, and
	\item{$\mathbf{u_v^c}$:}
		The only orientation of $u$, such that $\spt{u}{v}{3}$ is fully contained in the transmission cone of $u$.
\end{description}
\end{definition}

Notice that in each of the basic orientation of $u$ with respect to $v$, we have that $v$ lies in $u$'s cone. Therefore, for any assignment of basic orientation to $u$ (with respect to $v$) and any assignment of basic orientation to $v$ (with respect to $u$), the edge $\{u,v\}$ will be present in the resulting transmission graph.

Next, we prove three claims concerning the relationship between $\upt{u}{v}$ and $\upt{x}{y}$, where $\{u,v\}$ and $\{x,y\}$ are unordered pairs of points. 

\begin{figure}[ht]
	\centering
	\includegraphics[scale=0.5]{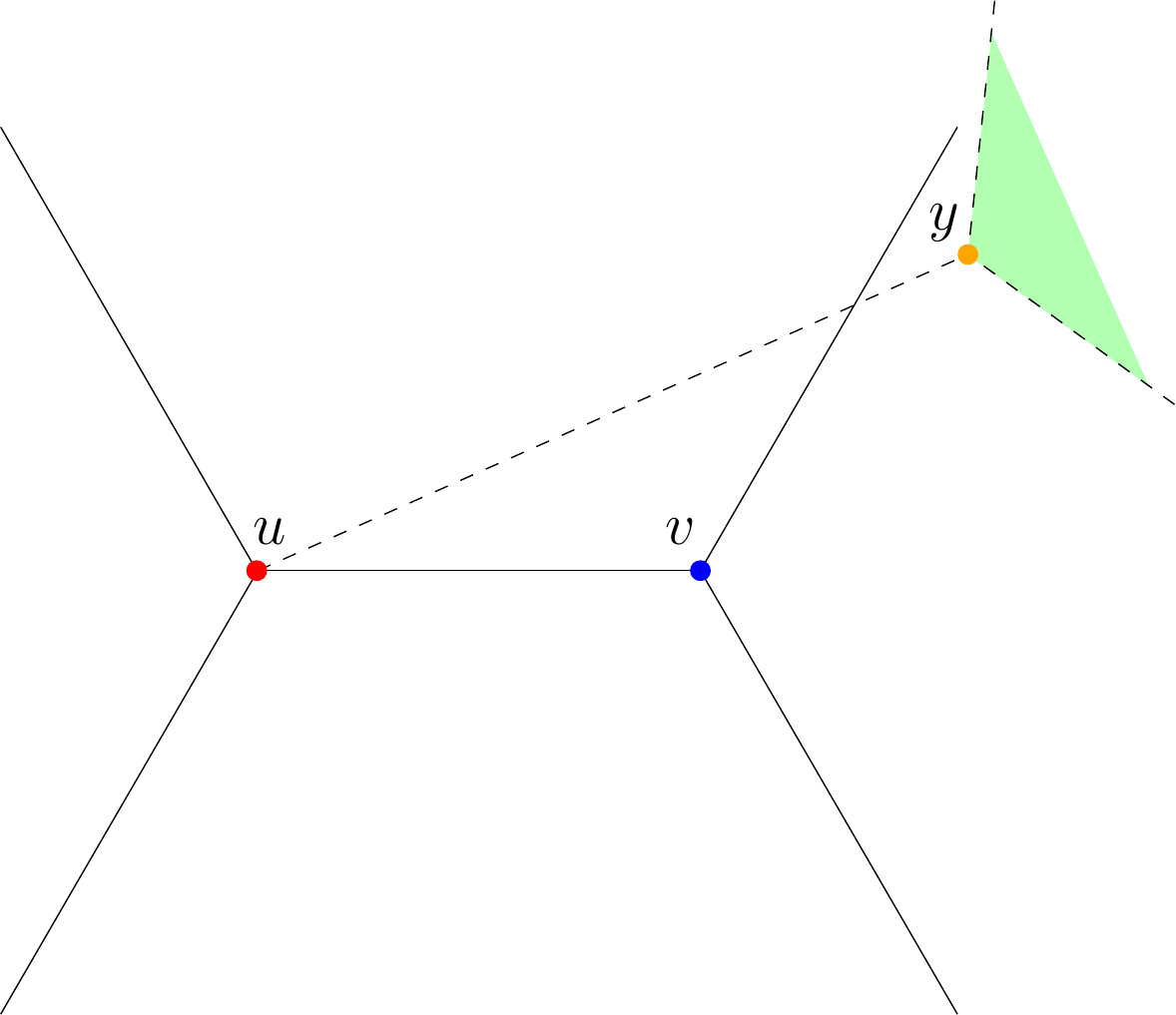}
	\caption{Proof of Claim~\ref{clm:1324}: If $y\in\spt{u}{v}{3}$ and $u\in\spt{x}{y}{3}$, then $x$ must lie in the green region.}
	\label{fig:claim1324}
\end{figure}

\begin{claim}
\label{clm:1324}
	Let $\{u,v\}$ and $\{x,y\}$ be two unordered pairs of points. If $x$ lies in one of the side regions of $\mathcal{P}_{\{u,v\}}$ and $y$ lies in the other, then both $u$ and $v$ lie in the union of the center regions of $\mathcal{P}_{\{x,y\}}$.
\end{claim}

\begin{proof}
\old{
	Assume for a contradiction that $x\in\spt{u}{v}{1}$, $y\in\spt{u}{v}{3}$, and $u\in\spt{x}{y}{3}$. For simplicity we also assume that the segment $\overline{uv}$ is horizontal, and that $y$ is not beneath the line containing it. We now notice that for the line segment $\overline{uy}$ $x$ must reside in the cone of angle $\frac{2\pi}{3}$ opposite to $\overline{uy}$ (See Figure~\ref{fig:claim1324}). This means a $\frac{\pi}{6}$ turn either left or right, and thus the cone of feasible locations of $x$ does not intersect with $\spt{u}{v}{1}$, contradicting the assumption.
}
	Assume, e.g., that $x\in\spt{u}{v}{1}$ and $y\in\spt{u}{v}{3}$. If $u$ is not in one of the center regions of $\mathcal{P}_{\{x,y\}}$, then it is in one of the side regions of $\mathcal{P}_{\{x,y\}}$. But, if $u\in\spt{x}{y}{1}$, then it is impossible that $y\in\spt{u}{v}{3}$, and if $u\in\spt{x}{y}{3}$, then it is impossible that $x\in\spt{u}{v}{1}$.
	Consider for example the latter case, i.e., $u\in\spt{x}{y}{3}$, and assume, without loss of generality, that the line segment 
	$\overline{uv}$ is horizontal, with $u$ to the left of $v$, and that $y$ is not below the line containing $\overline{uv}$ (see Figure~\ref{fig:claim1324}). Then, the requirement $u\in\spt{x}{y}{3}$ implies that $\spt{x}{y}{3}$ and the green region in the figure are disjoint (when viewed as open regions), which, in turn, implies that $x$ must lie in the green region. But this is impossible since the green region and $\spt{u}{v}{1}$ are disjoint.
\end{proof}

\old{
\stav{A possible claim that combines claims \ref{clm:23} and \ref{clm:32}. Also added a figure that illustrates the claims.}
{\color{olive}{
		\begin{claim}
			\label{clm:combined_claim}
			Let $\{u,v\}$ and $\{x,y\}$ be two unordered pairs of points, such that $x$ lies in one of the side regions of $\pt{u}{v}$, say in the one adjacent to $v$, and $y$ lies in one of the center regions of $\pt{u}{v}$. The following holds:
			\begin{enumerate}
				\item If $u$ lies in the side region adjacent to $x$, then so does $v$,
				\item If $v$ lies in the side region adjacent to $y$, then so does $u$.
			\end{enumerate}
			See Figure~\ref{fig:combined_claim}
		\end{claim}
}}

\begin{figure}[ht]
	\centering
	\begin{subfigure}{0.45\textwidth}
		\includegraphics[width=\textwidth, page=3]{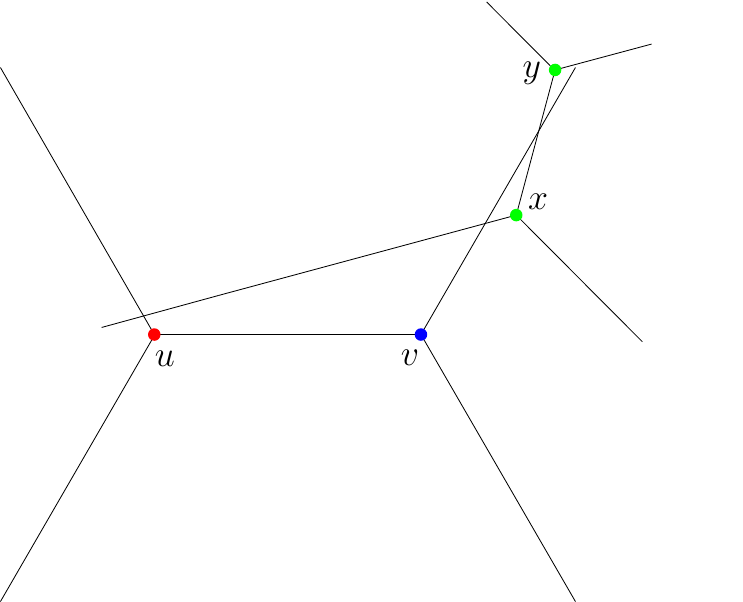}
		\label{fig:combined_claim_1}
		\subcaption{}
	\end{subfigure}
	\hfil
	\begin{subfigure}{0.45\textwidth}
		\includegraphics[width=\textwidth, page=2]{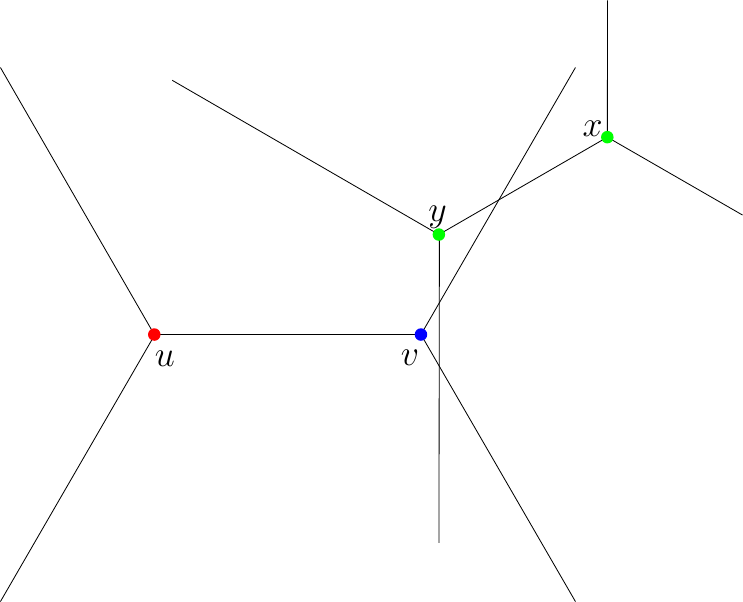}
		\label{fig:combined_claim_2}
		\subcaption{}
	\end{subfigure}
	\caption{An illustration of Claim~\ref{clm:combined_claim}. The first scenario depicted in the left sub-figure, and the second in the right.}
	\label{fig:combined_claim}
\end{figure}
}

\begin{figure}[ht]
	\centering
	\includegraphics[scale=0.5]{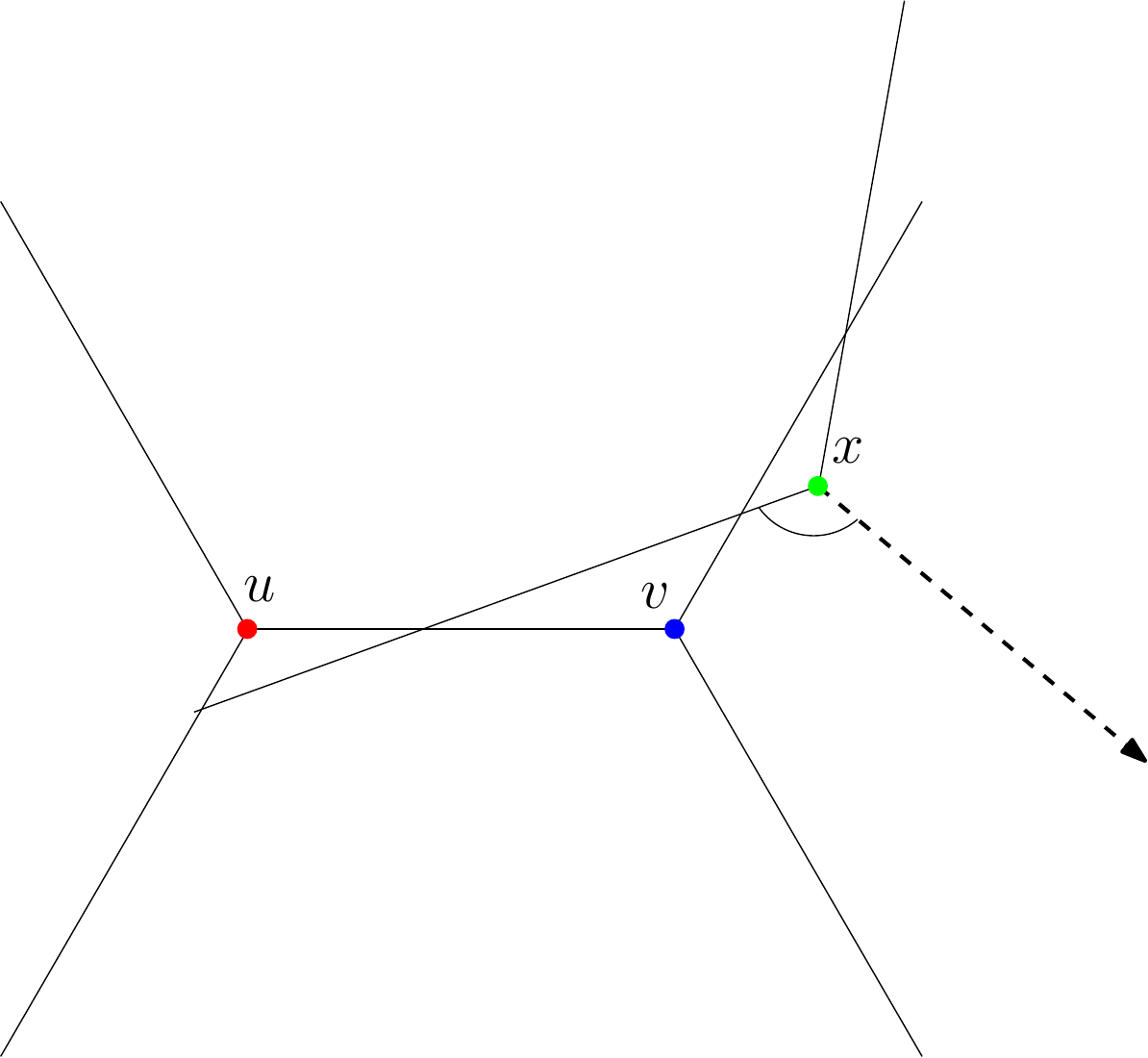}
	\caption{Proof of Claim~\ref{clm:23}: If $u\in\spt{x}{y}{1}$ but $v\notin\spt{x}{y}{1}$, then $y$, which is on the dashed ray emanating from $x$, is necessarily in $\spt{u}{v}{3}$.}
	\label{fig:claim23}
\end{figure}

\begin{claim}
\label{clm:23}
	Let $\{u,v\}$ and $\{x,y\}$ be two unordered pairs of points, such that $x$ lies in one of the side regions of $\upt{u}{v}$, say in the one adjacent to $v$, and $y$ lies in one of the center regions of $\upt{u}{v}$. Then, if $u$ lies in the side region adjacent to $x$, then so does $v$.
\end{claim}

\begin{proof}
\old{
	Assume for a contradiction and w.l.o.g that $x\in\spt{u}{v}{3}$, $y$ is in one of the center regions of $\pt{u}{v}$, $u\in\spt{x}{y}{1}$, but $v\notin\spt{x}{y}{1}$. We further assume w.l.o.g that	the segment $\overline{uv}$ is horizontal, and that $x$ is not below the line containing it. See Figure~\ref{fig:claim23} for an illustration.
	Since we know that $u$ and $v$ are in different parts of $\pt{x}{y}$, we know that the ray emanating from $x$ that separates $\spt{x}{y}{1}$ in which $u$ is located, from one of the center regions crosses $\overline{uv}$.
	Due to the definition of the partition $\pt{x}{y}$, we have that $y$ is then located on the line that passes through $x$, and creates an angle of $\frac{2\pi}{3}$ with $\overline{uv}$ (shown as the dashed arrow in Figure~\ref{fig:claim23}). Due to the structure of the partition $\pt{u}{v}$, this means that $y$ must reside in $\spt{u}{v}{3}$ as well. A contradiction to the assumption that $y$ is in one of the center regions.
}
	Assume that $u\in\spt{x}{y}{1}$ but $v\notin\spt{x}{y}{1}$. We show that this implies that $y\in\spt{u}{v}{3}$ --- a contradiction.
	Indeed, assume, without loss of generality, that the segment $\overline{uv}$ is horizontal, with $u$ to the left of $v$, and that $x$ is not below the line containing $\overline{uv}$ (see Figure~\ref{fig:claim23}).
	Since $u$ and $v$ are in different regions of $\pt{x}{y}$, we know that the border between $\spt{x}{y}{1}$ (in which $u$ resides) and $\spt{x}{y}{4}$ (in which $v$ resides) crosses $\overline{uv}$. But, this implies that the dashed ray emanating from $x$ is contained in $\spt{u}{v}{3}$, so $y$, which is somewhere on this ray, is in $\spt{u}{v}{3}$.
\end{proof}

\begin{figure}[ht]
	\centering
	\includegraphics[scale=0.5]{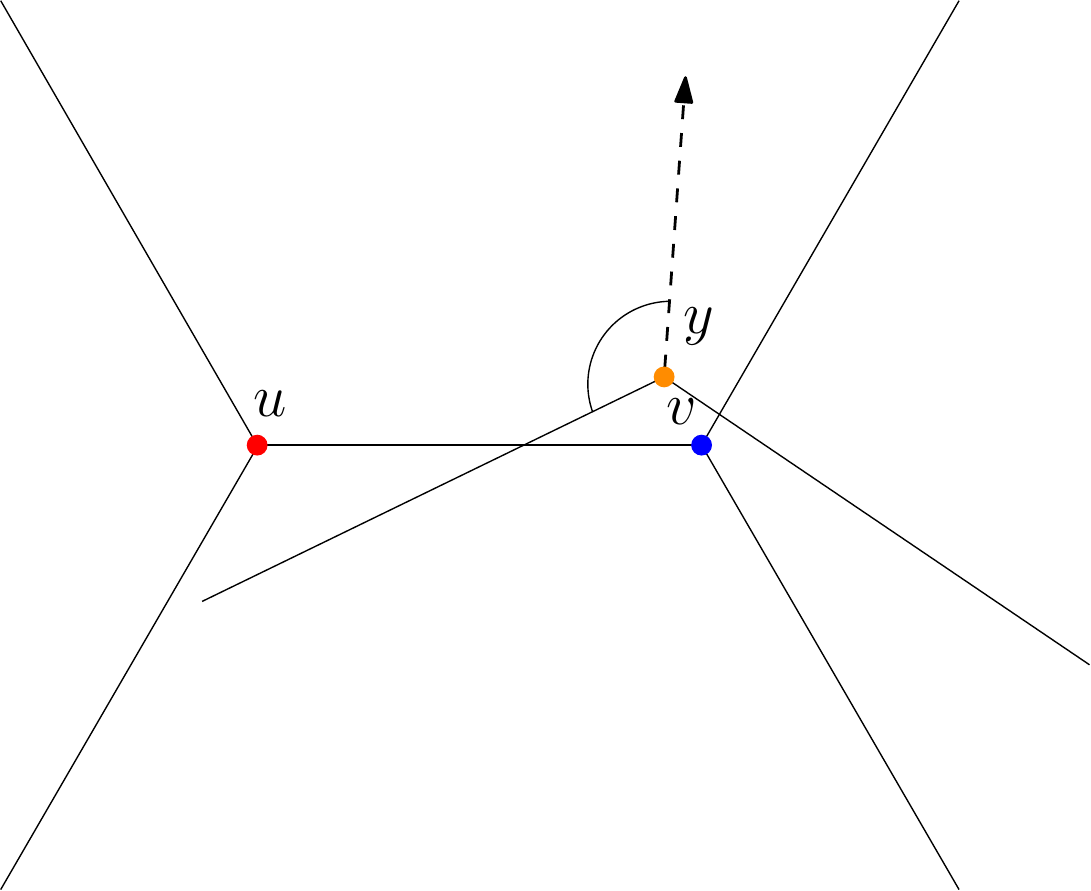}
	\caption{Proof of Claim~\ref{clm:32}:  If $v\in\spt{x}{y}{3}$ but $u\notin\spt{x}{y}{3}$, then $x$, which is on the dashed ray emanating from $y$, is necessarily in $\spt{u}{v}{2}$.}
	\label{fig:claim32}
\end{figure}

\begin{claim}
	\label{clm:32}
	Let $\{u,v\}$ and $\{x,y\}$ be two unordered pairs of points, such that $x$ lies in one of the side regions of $\mathcal{P}_{\{u,v\}}$, say in the one adjacent to $v$, and $y$ lies in one of the center regions of $\mathcal{P}_{\{u,v\}}$. Then, if $v$ lies in the side region adjacent to $y$, then so does $u$.
\end{claim}

\begin{proof}
\old{
	Assume for a contradiction and w.l.o.g that $x\in\spt{u}{v}{3}$, $y\in\spt{u}{v}{2}$, and $v\in\spt{x}{y}{3}$ but $u\notin\spt{x}{y}{3}$. We also assume that the segment $\overline{uv}$ is horizontal and that $u$ is to the left of $v$. We get that the boundary between $\spt{x}{y}{3}$ and one of the center regions of $\pt{x}{y}$ intersects $\overline{uv}$. See Figure~\ref{fig:claim32} for an illustration.
	
	Due to the definition of the partition $\pt{x}{y}$, we have that $x$ is then located on the line that passes through $y$, and creates an angle of $\frac{2\pi}{3}$ with $\overline{uv}$ (shown as the dashed arrow in Figure~\ref{fig:claim32}). Due to the structure of the partition $\pt{u}{v}$, this means that unless the boundary between $\spt{x}{y}{3}$ and $\spt{x}{y}{4}$ is parallel to $\overline{uv}$, $x$ cannot reside in $\spt{u}{v}{3}$. But, if the to lines are indeed parallel and intersecting as we have assumed, we have that they overlap, $u,v$ and $y$ are co-linear, and that either $y=v$ or $y\notin \spt{u}{v}{2}$ (since $y$ cannot be to the left of $v$ if $u,v$ and $y$ are collinear and $x\in\spt{u}{v}{3}$).
}
	Assume that $v\in\spt{x}{y}{3}$ but $u\notin\spt{x}{y}{3}$. We show that this implies that $x$ is in one of the center regions of $\pt{u}{v}$ --- a contradiction.
	Indeed, assume, without loss of generality, that the segment $\overline{uv}$ is horizontal, with $u$ to the left of $v$, and that $y \in \spt{u}{v}{2}$ (see Figure~\ref{fig:claim32}).
	Since $u$ and $v$ are in different regions of $\pt{x}{y}$, we know that the border between $\spt{x}{y}{3}$ (in which $v$ resides) and $\spt{x}{y}{4}$ (in which $u$ resides) crosses $\overline{uv}$. But, this implies that the dashed ray emanating from $y$ is contained in $\spt{u}{v}{2}$, so $x$, which is somewhere on this ray, is in $\spt{u}{v}{2}$.
\end{proof}

\section{Replacing an arbitrary path by a $\frac{2\pi}{3}$-tree}

Let $\{p_1,\ldots,p_n\}$ be a set of $n \ge 2$ points in the plane, and let $\Pi$ denote the polygonal path $(p_1,\ldots,p_n)$. The \emph{weight} of $\Pi$, $\omega(\Pi)$, is the sum of the lengths of the edges of $\Pi$, i.e., $\omega(\Pi) = \sum_{i=1}^{n-1} |p_ip_{i+1}|$. Let $X$ and $Y$ be the two natural matchings induced by $\Pi$, that is, $X = \{\{p_1,p_2\},\{p_3,p_4\},\ldots\}$ and $Y = \{\{p_2,p_3\},\{p_4,p_5\},\ldots\}$. Then, since $X \cap Y = \emptyset$, either $\omega(X)$ or $\omega(Y)$ is at most  $\omega(\Pi)/2$. Assume, without loss of generality, that  $\omega(X) \le \omega(\Pi)/2$. Moreover, assume for now that $n$ is even and that $X$ is a perfect matching.

In this section, we present an algorithm for replacing $\Pi$ by a $\frac{2\pi}{3}$-tree, ${\cal T}$, such that $\omega({\cal T}) \le 2\omega(\Pi)$ and, moreover, ${\cal T}$ is a 3-hop spanner of $\Pi$ (i.e., if there is an edge between $p$ and $q$ in $\Pi$, then there is a path consisting of at most three edges between $p$ and $q$ in ${\cal T}$).

Our algorithm assigns to each of the vertices $p$ of $\Pi$ an orientation, which is one of the three basic orientations of $p$ with respect to the vertex $q$ matched to $p$ in $X$.

In the subsequent description, we think of $X$ as a sequence (rather than a set) of edges.
Our algorithm consists of three phases.

\subsection{Phase I}

In the first phase of the algorithm, we iterate over the edges of $X$. When reaching the edge $\{p_i,p_{i+1}\}$, we examine it with respect to both its previous edge $\{p_{i-2},p_{i-1}\}$ and its next edge $\{p_{i+2},p_{i+3}\}$ in $X$. (The first edge is only examined w.r.t. its next edge, and the last edge is only examined w.r.t. its previous edge.) During the process, we either assign an orientation to one of $p_i,p_{i+1}$, to both of them, or to neither of them. In this phase, we only assign center orientations, i.e., $u_v^c$ or $v_u^c$, where $\{u,v\}$ is an edge in $X$.

\begin{figure}[ht]
\centering
\begin{subfigure}{0.45\textwidth}
	\includegraphics[width=\textwidth, page=4]{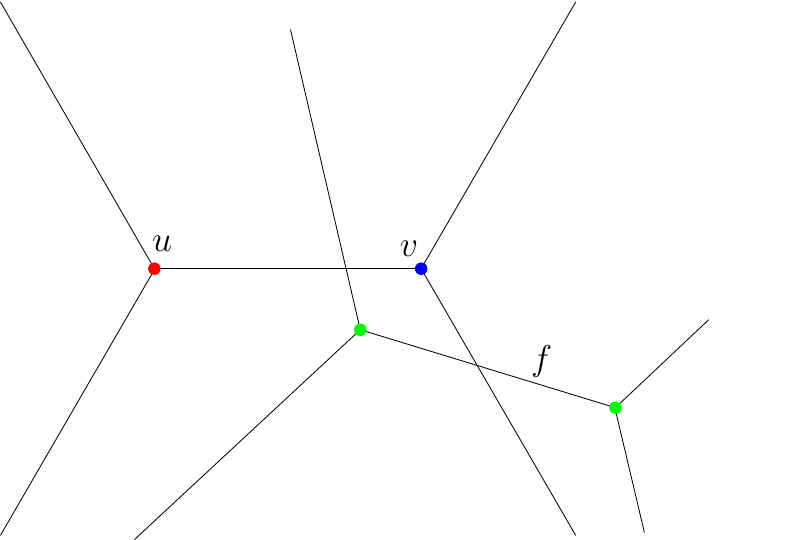}
	\label{fig:first phase 1}
\end{subfigure}
\hfil
\begin{subfigure}{0.45\textwidth}
	\includegraphics[width=\textwidth, page=2]{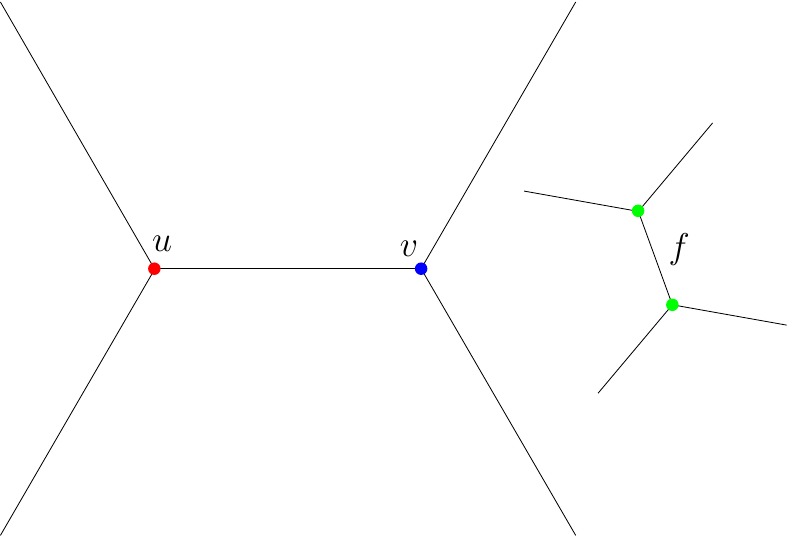}
	\label{fig:first phase 2}
\end{subfigure}
\caption{The conditions by which we assign $u$ the orientation $u_v^c$ due to $f=\{x,y\}$. {\bf Left:} $u$'s orientation is determined by the first condition. {\bf Right:} $u$'s orientation is determined by the second condition.}
\label{fig:first phase}
\end{figure}

Let $e = \{u,v\}$ be the edge that is being considered and let $f = \{x,y\}$ be one of its (at most) two neighboring edges.
We assign $u$ the orientation $u_v^c$ due to $f$ if one of the following conditions holds:
\begin{enumerate}
\item
One of $f$'s vertices is in $v$'s region (i.e., in the side region adjacent to $v$) and $u$ is in the region of the other vertex of $f$; see Figure~\ref{fig:first phase}~(left). 
\item
Both $x$ and $y$ are in $v$'s region; see Figure~\ref{fig:first phase}~(right). 
\end{enumerate}
Notice that it is possible that both conditions hold; see Figure~\ref{fig:first phase 3}. We say that $u$'s orientation was \emph{determined} by the second condition, only if the first condition does not hold; otherwise, we say that $u$'s orientation was \emph{determined} by the first condition.

Similarly, we assign $v$ the orientation $v_u^c$ due to $f$ if one of the conditions above holds, when $u$ is replaced by $v$.

\begin{figure}
\centering	
\includegraphics[width=0.6\textwidth, page=3]{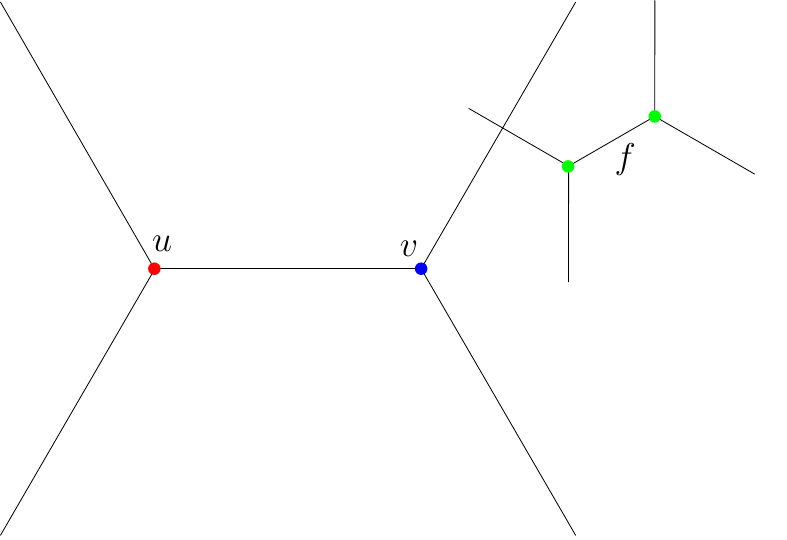}
\caption{Both conditions hold, but we say that $u$'s orientation is determined by the first condition.}
\label{fig:first phase 3}
\end{figure}

The following series of claims deals with the outcome of examining an edge $e$ with respect to a neighboring edge $f$.

\begin{claim}
	\label{clm:firstPhaseAtMostOne}
	The orientation of at most one of the vertices of edge $e=\{u,v\}$ is determined, when $e$ is examined with respect to a neighboring edge $f=\{x,y\}$. 
\end{claim}

\begin{proof}
	Assume that both $u$ and $v$ were oriented due to $f$ and consider the conditions responsible for it, so as to reach a contradiction. If the orientation of one of the vertices, say $u$, was determined by the second condition, then neither of the conditions can apply to $v$, since both conditions require that at least one of $f$'s vertices is in $u$'s region. If, however, the orientation of both $u$ and $v$ was determined by the first condition, then, without loss of generality, $x$ is in $u$'s region and $y$ is in $v$'s region, and by Claim~\ref{clm:1324} we conclude that $u$ and $v$ are in the center regions of $\pt{x}{y}$, implying that neither of the vertices of $e$ was oriented due to $f$.
\end{proof} 

\begin{claim}
	\label{clm:firstPhaseFirstRuleSymmetry}
	If the orientation of a vertex of edge $e=\{u,v\}$ is determined by the first condition, when $e$ is examined with respect to a neighboring edge $f=\{x,y\}$, then the orientation of a vertex of $f$ is determined by the first condition, when $f$ is examined with respect to $e$, and these two vertices induce an edge of the transmission graph.
\end{claim}

\begin{proof}
	Assume that, e.g., $u$'s orientation is determined by the first condition (i.e., $u$ is assigned the orientation $u_v^c$), when $e$ is examined with respect to $f$. This means that there is a vertex of $f$, say $x$, that is in $v$'s region, and that $u$ is in $y$'s region.
	Now, when we proceed to examine the edge $f$ with respect to $e$, we find that $u$ is in $y$'s region and $x$ is in $v$'s region, so by the first condition we assign $x$ the orientation $x_y^c$.
	
	It remains to show that $u$ and $x$ induce and edge of the transmission graph. Indeed, $x$ is in the transmission cone of $u$, since $x$ is in $v$'s region and $u$'s cone contains $v$'s region. Similarly, $u$ is in the transmission cone of $x$, since $u$ is in $y$'s region and $x$'s cone contains $y$'s region. 
\end{proof}

\begin{claim}
	If the orientation of a vertex of edge $e=\{u,v\}$ is determined by the second condition, when $e$ is examined with respect to a neighboring edge $f=\{x,y\}$, then neither of $f$'s vertices is assigned an orientation due to $e$. 
\end{claim}

\begin{proof}
	If the orientation of, e.g., $u$ is determined by the second condition, when $e$ is examined with respect to $f$, then $u$ is in one of the center regions of $\upt{x}{y}$. Therefore, when $f$ is examined with respect to $e$, the only condition that may hold is the first one. But if it does, then by Claim~\ref{clm:firstPhaseFirstRuleSymmetry}, the orientation of $u$ is determined by the first condition, contrary to our assumption. We conclude that if the orientation of a vertex of $e$ is determined by the second condition, then neither of $f$'s vertices is assigned an orientation due to $e$.
\end{proof}

\subsection{Phase II}
\label{sec:phase2}

After completing the first phase, in which we iterated over the edges of $X$ only once (i.e., a single round), we proceed to the second phase, in which we iterate over the edges of $X$ again and again (i.e., multiple rounds). The second phase ends only after a full round is completed, in which no vertex is assigned an orientation. 
In a single round, we iterate over the edges of $X$, and for each pair of consecutive edges $e=\{u,v\}$ and $f=\{x,y\}$, where $e$ precedes $f$, we assign orientations to the vertices of $e$ and $f$, subject to the four rules listed below.  

\begin{description}
	
	\item{\bf No reorienting:}
	The orientation of a vertex is unmodifiable; that is, once the orientation of a vertex has been fixed (possibly already in the first phase), it cannot be changed. 
    
    \item{\bf Center orientation:}
    A non-center orientation to a vertex $u$ of an edge $e$ is allowed, only if $u$ is the second vertex of $e$ to be assigned an orientation. Thus, if $u$ is the first vertex of $e$ to be assigned an orientation, then $u$ must be assigned a center orientation.
  
    \item{\bf Edge creation:}
    Every operation that is performed must result in the creation of an edge of the transmission graph. This is achieved either by assigning orientations to two vertices simultaneously, or by orienting a vertex towards an already oriented vertex.
    
    \item{\bf No double tapping:}
    If one of $e$'s vertices was already oriented due to $f$, where $f$ is one of $e$'s neighboring edges, then the other vertex of $e$ will not be oriented due to $f$. 
    
\end{description}

Notice that in this phase, unlike the previous one, the orientation decisions that we make when examining an edge $e$ with respect to the next edge $f$, also depend on the orientations that some of the vertices of these edges may already have, and not only on the relative positions of these vertices.

At this point, we could have proceeded directly to the third phase, since for the purpose of correctness we do not need to elaborate on the types of operations that are performed in the second phase. However, for the sake of clarity, we illustrate below several types of operations that are performed during the second phase.

\begin{figure}[ht]
	\centering
	\begin{subfigure}[Left]{0.48\textwidth}
		\includegraphics[width = \textwidth, page = 3]{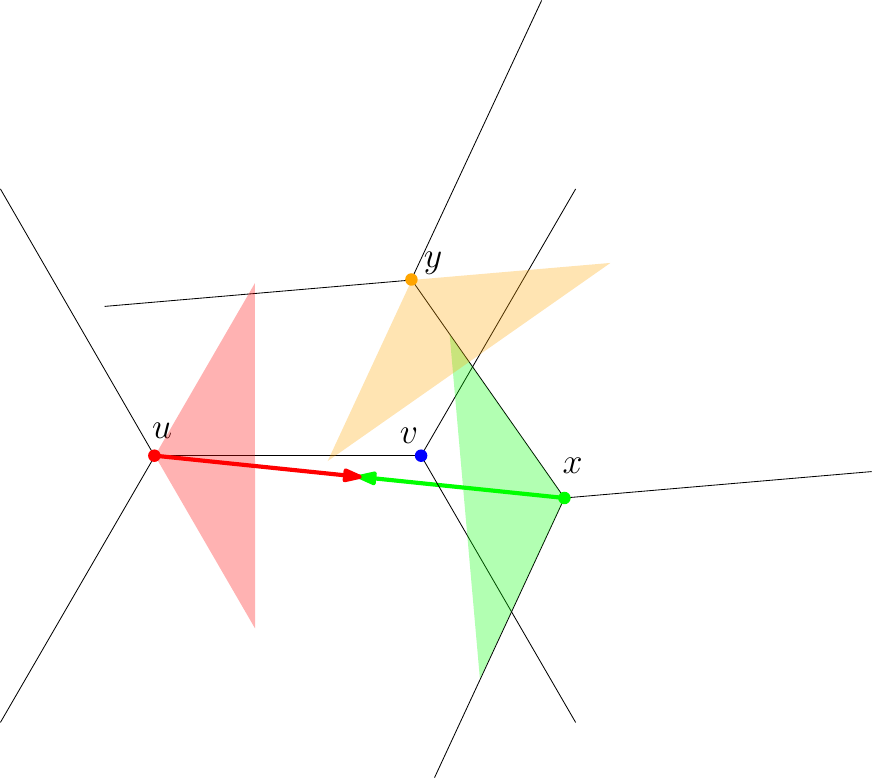}
		\label{fig:secondPhase1}
	\end{subfigure}
	\hfil
	\begin{subfigure}[Right]{0.48\textwidth}
		\includegraphics[width = \textwidth, page = 4]{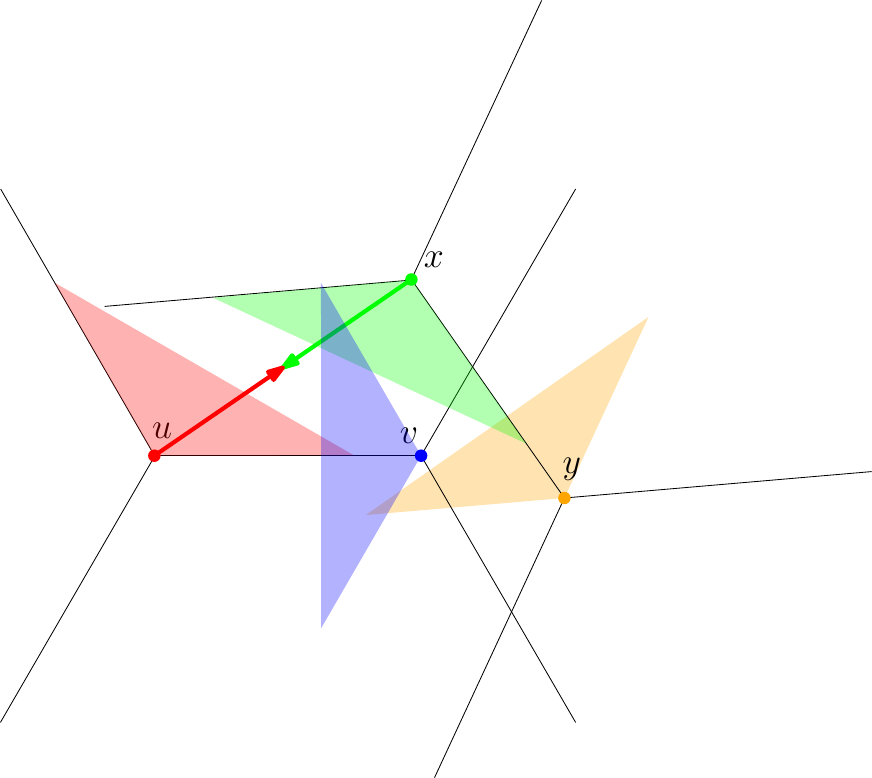}
		\label{fig:secondPhase2}
	\end{subfigure} 
	\caption{{\bf Left:} If $y$ was already oriented, $u$ and $x$ can be oriented towards each other, assigning $x$ a non-center orientation (first scenario). If $u$ was also already oriented, $x$ can be oriented towards $u$ (second scenario). {\bf Right:} If both $v$ and $y$ were already oriented, $u$ and $x$ can be oriented towards each other, assigning them non-center orientations (third scenario).}
	\label{fig:secondPhase} 
\end{figure}

\begin{itemize}
	\item
	$u$ is assigned the orientation $u_v^c$ and $x$ is assigned the orientation $x_y^u$ or $x_y^d$, to establish the edge $\{u,x\}$ of the transmission graph; see Figure~\ref{fig:secondPhase}~(left).\\ Precondition: $y$ is already oriented.
	
	\item
	$x$ is assigned the orientation $x_y^u$ or $x_y^d$, to establish the edge $\{u,x\}$ of the transmission graph, where $u$ was previously assigned the orientation $u_v^c$; see Figure~\ref{fig:secondPhase}~(left).\\
	Precondition: $y$ is already oriented.
	
	\item
	$u$ is assigned the orientation $u_v^u$ or $u_v^d$ and $x$ is assigned the orientation $x_y^u$ or $x_y^d$, to establish the edge $\{u,x\}$ of the transmission graph; see Figure~\ref{fig:secondPhase}~(right).\\
	Precondition: $v$ and $y$ are already oriented.
\end{itemize}

\old{ 
\begin{figure}[ht]
	\centering
	\includegraphics[scale= 0.7, page = 8]{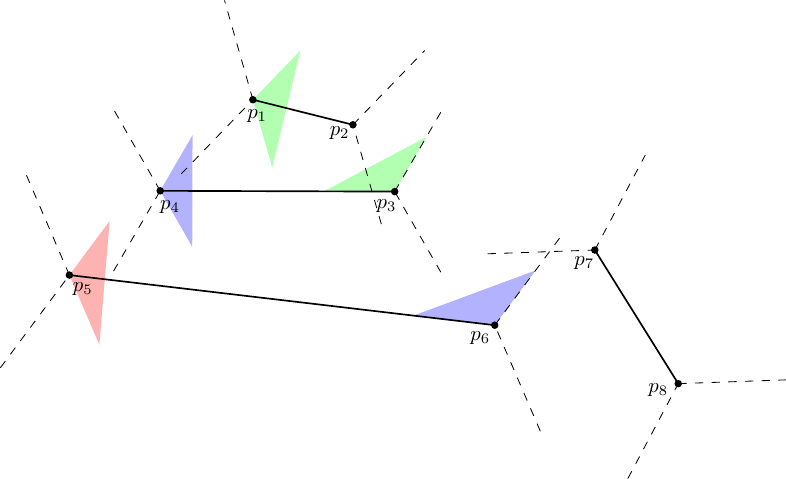}
	\caption{In the first phase, the vertex $p_5$ is assigned a center orientation (depicted by the red triangle) due to the edge $\{p_7, p_8\}$ (second condition). In the first round of the second phase, $p_4$ and $p_6$ are oriented towards each other (the blue triangles), and in the second round of the second phase $p_1$ and $p_3$ are oriented towards each other (the green triangles).}
	\label{fig:secondPhaseExample} 
\end{figure}
}

In Figure~\ref{fig:complete_example}(a-b) one can see the result of applying the first two phases to the path $(p_1,\ldots,p_8)$, i.e., to the sequence of edges $X = (\{p_1,p_2\},\{p_3,p_4\},\{p_5,p_6\},\{p_7,p_8\})$.  

\subsection{Phase III}

In this phase we perform one final round, in which we orient all the vertices that were not yet oriented. More precisely, we iterate over the edges of $X$, considering each edge $e$ with respect to the next edge $f$. When considering $e$, we orient its vertices that were not yet oriented, so that once we are done with $e$, both $e$ itself and an edge connecting $e$ and $f$ are present in the transmission graph that is being constructed.      

When considering the edge $e=\{u,v\}$ with respect to the next edge $f=\{x,y\}$, we know (by induction) that there already exists a transmission edge connecting $e$ and the previous edge, so at most one of $e$'s vertices was not yet oriented. If both vertices of $e$ were already oriented, then either there already exists a transmission edge connecting $e$ and $f$, or not. In the former case, proceed to the next edge of $X$ (i.e., to $f$), and in the latter case, orient a vertex of $f$ that was not yet oriented (there must be such a vertex), to obtain a transmission edge between $e$ and $f$. We prove below that this is always possible.

If only one of $e$'s vertices was already oriented, then let, e.g., $u$ be the one that is not yet oriented. 
Now, if there already exists a transmission edge connecting $e$ and $f$ (i.e., $v$ is connected to both the previous and the next edge of $e$), then assign $u$ the orientation $u_v^c$ (ensuring that $e$ is a transmission edge). Otherwise, if one can assign an orientation to $u$, so that a transmission edge is created between $u$ and an already oriented vertex of $f$, then do so. If this is impossible, then orient $u$ and a vertex of $f$ that was not yet oriented (there must be such a vertex), to obtain a transmission edge between $e$ and $f$. We prove below that this is always possible.

\old{
\begin{figure}[ht]
	\centering
	\includegraphics[scale= 0.7, page = 10]{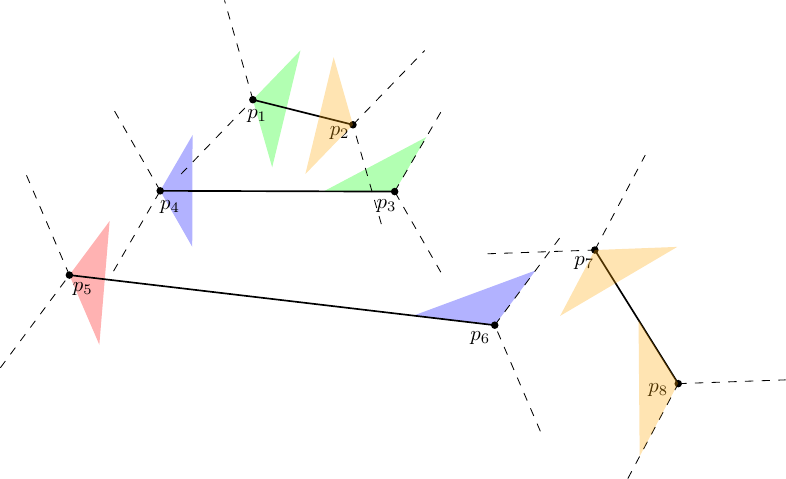}
	\caption{In the third phase, $p_2$, $p_7$, and $p_8$ are assigned orientations (the yellow triangles.). Notice that $p_8$ is oriented to obtain the transmission edge $\{p_5,p_8\}$.}
	\label{fig:thirdPhaseExample}
\end{figure}
}

In Figure~\ref{fig:complete_example}(c) one can see the result of applying the third phase to the sequence of edges $X = (\{p_1,p_2\},\{p_3,p_4\},\{p_5,p_6\},\{p_7,p_8\})$ (following the application of the first and second phases).  

\subsection{Correctness}

We first consider the more interesting case, where (i) one of the vertices of $e$, say $u$, is not yet oriented, (ii) there is no transmission edge between $e$ and $f$, and (iii) it is impossible to orient $u$ so that a transmission edge is created between $u$ and an already oriented vertex of $f$. In this case, we need to prove that at least one of $f$'s vertices is not yet oriented and that it is possible to orient both $u$ and such a vertex of $f$ to obtain a transmission edge between $e$ and $f$.

We begin by showing the if both of $f$'s vertices were already oriented, then either assumption (ii) or assumption (iii) does not hold. 
Indeed, by Claim~\ref{clm:firstPhaseAtMostOne} and the \emph{No double tapping} rule of the second phase, one $f$'s vertices, say $x$, was oriented due to $e$. Now, if $x$ was oriented during the first phase, then we distinguish between two cases according to the condition by which the orientation of $x$ was determined.
\begin{description}
	\item{\textbf{$x$'s orientation was determined by the first condition.}} In this case, by Claim~\ref{clm:firstPhaseFirstRuleSymmetry}, the edge $\{v,x\}$ is already in the transmission graph. In more detail, since $u$ is not yet oriented, we must have that $x\in\spt{u}{v}{1}$ and $v\in\spt{x}{y}{3}$.
			
	\item{\textbf{$x$'s orientation was determined by the second condition.}} In this case, both $u$ and $v$ are in $y$'s region and $x$ is in one of the center regions of $\upt{u}{v}$. So, by orienting $u$ appropriately, one can obtain the transmission edge $\{u,x\}$.
			
\end{description}
If, however, $x$ was oriented during the second phase, then by the \emph{Edge creation} rule, an edge connecting $e$ and $f$ was already created.

We thus conclude that at least one of $f$'s vertices is not yet oriented. We now consider, separately, the case where only one of $f$'s vertices is not yet oriented and the case where both vertices of $f$ are not yet oriented.

\noindent
{\bf Only one of $f$'s vertices is not yet oriented.}
Assume, without loss of generality, that $y$ is the vertex of $f$ that is already oriented. If $y$ was oriented due to $e$, then by replacing $x$ with $y$ in the proof above, we get that either assumption (ii) or assumption (iii) does not hold. Therefore, we assume that $y$ was oriented due to the edge following $f$, which implies that $y$ was oriented in the first or second phase. Now, if $u$ and $x$ can be oriented to obtain the transmission edge $\{u,x\}$, then we are done. Otherwise, $u \in \spt{x}{y}{1}$ or $x \in \spt{u}{v}{1}$. We consider these cases below and show, for both of them, that a transmission edge between $e$ and $f$ can still be created.
\begin{description}
	\item{$\mathbf{u\in\spt{x}{y}{1}}$:} Notice that since $x$ is not yet oriented and $y$ was oriented in the first or second phase, $y$'s orientation is necessarily $y_x^c$. We consider each of the possible locations of $v$ in $\pt{x}{y}$, and show that regardless of $v$'s location a transmission edge can be created.
	\begin{enumerate}
		\item If $v\in\spt{x}{y}{1}$, then $y$ was oriented due to $e$ during the first phase --- contradiction.
			
		\item If $v\in\spt{x}{y}{3}$, then, by Claim~\ref{clm:1324}, $x$ and $y$ are in the center regions of $\pt{u}{v}$, which allows us to orient $u$ towards $y$ to create the transmission edge $\{u,y\}$.
			
		\item If $v$ is in one of the center regions of $\pt{x}{y}$, then we apply Claim~\ref{clm:23} to show that we can orient $x$ towards $v$ to create the transmission edge $\{v,x\}$. Indeed,  since (by assumption~(iii)) we cannot orient $u$ to create the transmission edge $\{u,y\}$, we know that $y \in\spt{u}{v}{1}$. So by Claim~\ref{clm:23}, we get that $x \in \spt{u}{v}{1}$. Therefore, since both $x$ and $y$ are in $u$'s region, $v$'s orientation was determined by the second condition during the first phase, and $x$ can be oriented towards $v$ to create the transmission edge $\{v,x\}$.
	\end{enumerate}

	\item{$\mathbf{x\in\spt{u}{v}{1}}$:} 
	We first observe that if it is possible to create a transmission edge between $v$ and $x$ (i.e., $v \not \in \spt{x}{y}{1}$), then it is possible to do so by assigning $v$ a center orientation (since $x\in\spt{u}{v}{1}$), and we would have created the edge $\{v,x\}$ (by assigning $v$ a center orientation and $x$ an appropriate orientation) in the second phase, as $y$ was oriented in the first or second phase. We assume therefore that it is impossible to create a transmission edge between $v$ and $x$, which implies that $v \in \spt{x}{y}{1}$.
	
	We now show that regardless of the location of $y$ in $\pt{u}{v}$, we get that $v \not \in \spt{x}{y}{1}$ --- contradiction. 
	\begin{enumerate}
		\item If $y\in\spt{u}{v}{3}$, then, by Claim~\ref{clm:1324}, $v$ is in a center region of $\pt{x}{y}$.
			
		\item If $y\in\spt{u}{v}{1}$, 
		then an edge between $v$ and $y$ was created in the first phase (i.e., the orientations of both $v$ and $y$ were determined by the first condition of the first phase).
			
		\item If $y$ is in one of the center regions of $\pt{u}{v}$, say $y \in \spt{u}{v}{2}$, then, by Claim~\ref{clm:23} and using the assertion that $v\in\spt{x}{y}{1}$, we get that  $u\in\spt{x}{y}{1}$ as well. Therefore, $y$ was assigned a center orientation in the first phase due to $e$, in contradiction to our assumption.
	\end{enumerate}
		
\end{description}

\noindent
{\bf Both vertices of $f$ are not yet oriented.}
If $x,y \in \spt{u}{v}{1}$, then $v$'s orientation was determined by the second condition in the first phase (since if it were determined by the first condition, then we would already have an edge between $e$ and $f$). Therefore, $v$'s orientation is $v_u^c$ and $v$ is in one of the center regions of $\pt{x}{y}$, and we orient either $x$ or $y$ towards $v$ to create a transmission edge between $e$ and $f$.
 
Assume, therefore, that at least one of $f$'s vertices, say $x$, is not in $\spt{u}{v}{1}$. Now, if $u \notin \spt{x}{y}{1}$, then we orient $u$ and $x$ towards each other to create the edge $\{u,x\}$. So assume, in addition, that $u \in \spt{x}{y}{1}$. Under these assumptions, we show that regardless of the location of $x$ in $\pt{u}{v}$, $y \notin \spt{u}{v}{1}$, so $u$ and $y$ can be oriented towards each other to create the transmission edge $\{u,y\}$. 

\begin{enumerate}
		\item If $x$ is in one of the center regions of $\pt{u}{v}$, say $x \in \spt{u}{v}{2}$, then $y \notin \spt{u}{v}{1}$. Since, $y \in \spt{u}{v}{1}$, $x \in \spt{u}{v}{2}$ and $u \in \spt{x}{y}{1}$ implies (by Claim~\ref{clm:32}) that $y$ was already oriented in the first phase.
		
		\item If $x \in \spt{u}{v}{3}$, then again $y \notin \spt{u}{v}{1}$. Since, $y \in \spt{u}{v}{1}$ and $x \in \spt{u}{v}{3}$ implies (see Claim~\ref{clm:1324}) that $u$ is in one of the center regions of $\pt{x}{y}$, contradicting the assumption $u\in\spt{x}{y}{1}$.
\end{enumerate}

We now tend to the case where both vertices of $e$ are already oriented, but there is no transmission edge between $e$ and $f$.
We first notice that this means that one of the vertices of $e$, say $u$, was oriented due to $f$. Moreover, $u$'s orientation was determined by the second condition in the first phase, since otherwise an edge connecting $e$ and $f$ would already exist in the transmission graph. Next, we notice that at least one of the vertices of $f$ was not yet oriented, since if both were oriented, then, again, one of them was oriented due to $e$ and its orientation was determined by the second condition in the first phase. But, this implies that the first condition applies to both $u$ and this vertex of $f$ and that a transmission between them was already exists.

Now, since $u$'s orientation was determined by the second condition in the first phase, we know that it is in one of the center regions of $\pt{x}{y}$. We can therefore orient the vertex of $f$ that is not yet oriented towards $u$ to create the required transmission edge.  

\begin{figure}[ht]
	\begin{subfigure}{0.48\textwidth}
		\includegraphics[width=\textwidth, page=4]{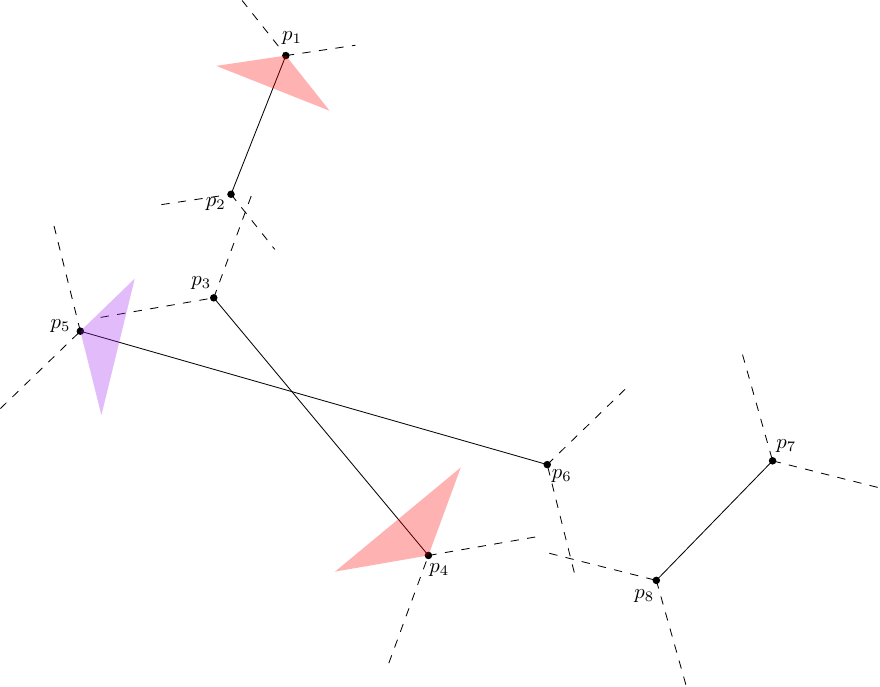}
		\subcaption{After Phase I.}
	\end{subfigure}
	\hfill
	\begin{subfigure}{0.48\textwidth}
		\includegraphics[width=\textwidth, page=5]{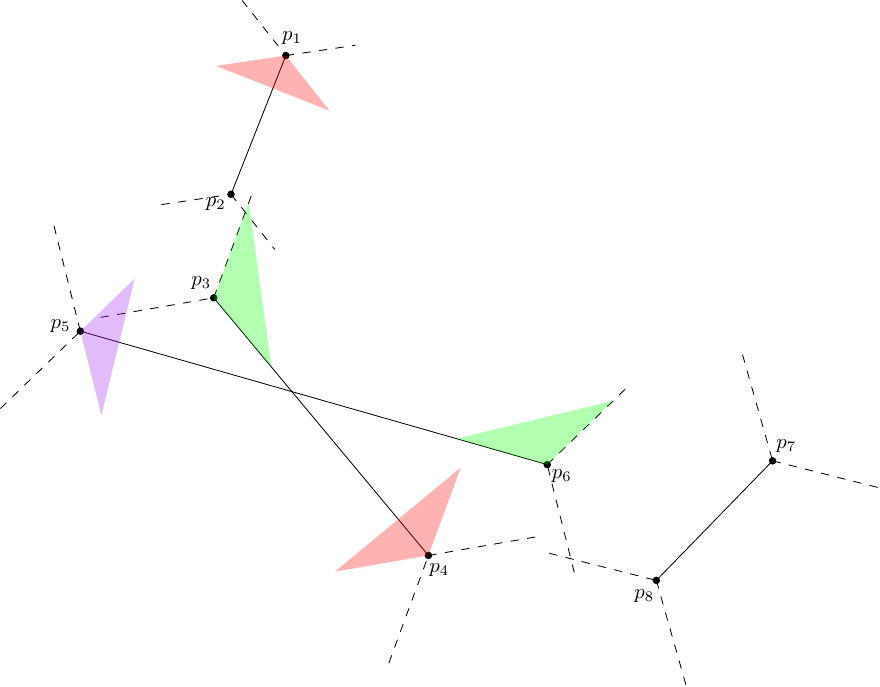}
		\subcaption{After Phase II.}
	\end{subfigure}
	
	\vspace{0.2cm}
	
	\begin{subfigure}{0.48\textwidth}
		\includegraphics[width=\textwidth, page=6]{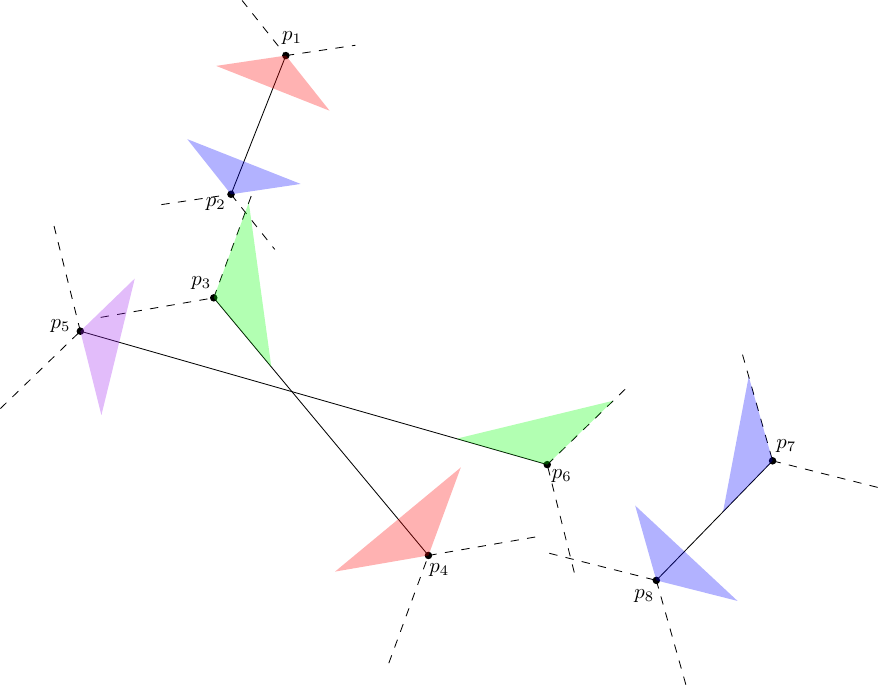}
		\subcaption{After Phase III.}
	\end{subfigure}
	\hfill
	\begin{subfigure}{0.48\textwidth}
		\includegraphics[width=\textwidth, page=7]{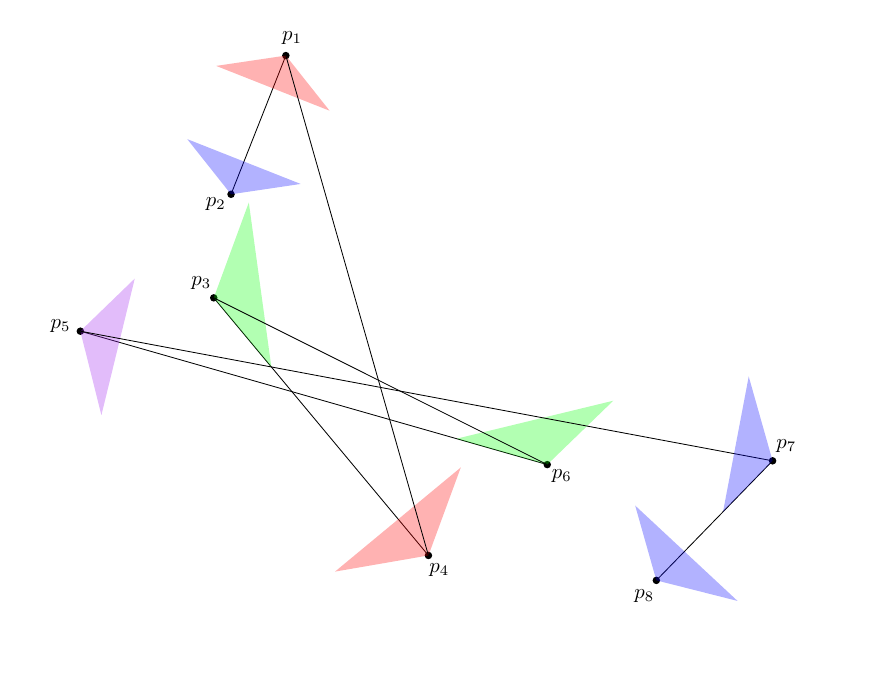}
		\subcaption{The $\frac{2\pi}{3}$-ST that is obtained.}
	\end{subfigure}
	\caption{A complete example. In the first phase, the vertex $p_1$ ($p_4$) is assigned a center orientation (depicted by the red triangle) due to the edge $\{p_3, p_4\}$ ($\{p_1, p_2\}$) (first condition). The vertex $p_5$ is assigned a center orientation (depicted by the violet triangle) due to the edge $\{p_7, p_8\}$ (second condition). In the first round of the second phase, $p_3$ and $p_6$ are oriented towards each other, when considering the consecutive edges $\{p_3,p_4\}, \{p_5,p_6\}$ (the green triangles). In the third phase, the vertices $p_2$, $p_7$, and $p_8$ are assigned orientations (the blue triangles). The bottom-right figure shows the $\frac{2\pi}{3}$-ST that is obtained.}
	\label{fig:complete_example}
\end{figure}

At this point, the edge set of our transmission graph $G$ contains $X$ and at least one edge, for each pair $e,f$ of consecutive edges of $X$, connecting a vertex of $e$ and a vertex of $f$. Let ${\cal T}$ be the graph obtained from $G$ by leaving only one (arbitrary) edge, for each pair of consecutive edges of $X$. Then, ${\cal T}$ is a $\frac{2\pi}{3}$-spanning tree of $P$ (see Figure~\ref{fig:complete_example}(d)). Denote by $Y'$ the set of edges of ${\cal T}$ between (vertices of) consecutive edges of $X$. Then, $\omega({\cal T}) = \omega(X) + \omega(Y') \leq \omega(X) + (2\omega(X) + \omega(Y)) =  \omega(\Pi) + 2\omega(X) \leq 2\omega(\Pi)$.
Moreover, ${\cal T}$ is a 3-hop spanner of $\Pi$, in the sense that if $\{p,q\}$ is an edge of $\Pi$, then there is a path between $p$ and $q$ in ${\cal T}$ consisting of at most 3 edges.

\noindent
{\bf The non-perfect case.}
It was convenient to assume that $X$ is a perfect matching, but it is possible of course that it is not. More precisely, if $n$ is odd, then $|X| = \lfloor{\frac{n}{2}}\rfloor$ and either $p_1$ or $p_n$ remain unmatched, and if $n$ is even, then either $|X| = \frac{n}{2}$ or $|X| = \frac{n}{2}-1$, where in the latter case both $p_1$ and $p_n$ remain unmatched.
However, it is easy to deal with the case where $X$ is not a perfect matching, by converting it to the case where it is. Roughly, for each unmatched point $p \in P$, we add a new point $p'$ to $P$ and add the edge $e=\{p,p'\}$ to $X$. We then apply the algorithm as described above. 

We now provide a more detailed description of this reduction. Let $f=\{u,v\}$ be the edge of $X$ adjacent to $p$ (e.g., if $p=p_1$, then $f=\{p_2,p_3\}$). We draw $p'$ close enough to $p$, ensuring that both points lie in the same region of $\pt{u}{v}$, and add the edge $e = \{p,p'\}$ to $X$. We now apply the algorithm above to the perfect matching $X$, and consider the resulting transmission graph $G$. If $G$ contains an edge between $p$ and a vertex of $f$, then simply remove the point $p'$ from $G$. Otherwise, $G$ must contain an edge between $p'$ and a vertex of $f$, say $u$. In this case, we orient $p$ towards $u$, thus creating the transmission edge $\{p,u\}$ (since $p$ is also in $u$'s transmission cone). Finally, we remove the point $p'$.

\noindent
{\bf Running time.}
The first and third phases of the algorithm each consist of a single round, whereas the second phase consists of $O(n)$ rounds. In each round we traverse the edges of $X$ from first to last and spend $O(1)$ time at each edge. Thus, the running time of the first and third phases is $O(n)$, whereas the running time of the second phase is $O(n^2)$. We show below that the quadratic bound on the running time is due to our desire to keep the description simple, and that by slightly modifying the second phase we can reduce its running time to $O(n)$.
The modification is based on the observation that beginning from the second round, an operation is performed when considering the pair $e_i, e_{i+1}$ of edges of $X$ (i.e., a transmission edge between them is created) if (i) an operation was performed in the previous round when considering $e_{i+1}$ and $e_{i+2}$, or (ii) an operation was performed in the current round when considering $e_{i-1}$ and $e_i$ (or both).

Using this observation, we prove that two rounds are sufficient. Specifically, in the first round, we traverse the edges of $X$ from first to last, i.e., a \emph{forward} round, and in the second round, we traverse the edges of $X$ from last to first, i.e., a \emph{backward} round. In both rounds, in each iteration we consider the current edge and the following one, and check whether an operation can be performed (i.e., a transmission edge can be created), under the four rules listed in Section~\ref{sec:phase2}. We refer to such an operation as a \emph{legal} operation.  

We now prove that once we are done, no legal operation can be performed when considering a pair of adjacent edges of $X$.
Indeed, let $g=\{p_{i},p_{i+1}\}$, $f=\{p_{i+2},p_{i+3}\}$, $e$, and $d$ be four consecutive edges of $X$, and assume that after the backward round, one can still perform a legal operation when considering the pair $e$ and $f$. 
Then, the operation became legal after an operation was performed when considering the pair $f$ and $g$. Since, if it became legal after an operation was performed when considering the pair $d$ and $e$, then we would have performed it during the backward round. However, by our assumption, no operation was performed during the backward round when considering the pair $e$ and $f$, and therefore no operation was performed in this round when considering the pair $f$ and $g$ --- contradiction.

The following theorem summarizes the main result of this section.
\begin{theorem}
\label{thm:path2approx}
	Let $P = \{p_1,\ldots,p_n\}$ be a set of points in the plane, and let $\Pi$ denote the polygonal path $(p_1,...,p_n)$. Then, one can construct, in $O(n)$-time, a $\frac{2\pi}{3}$-spanning tree ${\cal T}$ of $P$, such that (i) $\omega({\cal T}) \le 2\omega(\Pi)$, and (ii) ${\cal T}$ is a 3-hop spanner of $\Pi$.
\end{theorem}

\begin{corollary}
	Let $P = \{p_1,\ldots,p_n\}$ be a set of $n$ points in the plane. Then, one can construct in $O(n\log n)$-time a $\frac{2\pi}{3}$-ST\,  ${\cal T}$ of $P$, such that $\omega({\cal T}) \le 4\omega(\mathrm{MST}(P))$.  
\end{corollary}

\old{
\section{conclusion} 
Ignoring the first stage of the algorithms of Aschner and Katz~\cite{AschnerK17} and Biniaz et al.~\cite{BiniazBLM20}, in which a (non-crossing) spanning path of $P$ is computed (in either $O(n\log n)$ time or $O(n^3)$ time respectively), the running time of these algorithms is $O(n)$, while the running time of our algorithm is $O(n^2)$. The reason for this is ... 
}

\bibliography{refs}

\end{document}